\documentclass{article}
\usepackage{graphicx, color} 
\usepackage{float}
\usepackage{adjustbox} 
\usepackage[margin=1.4in]{geometry}

\usepackage[onehalfspacing]{setspace}
\graphicspath{{figures/}}
\usepackage{amsthm}
\usepackage{xcolor}
\usepackage{cancel}
\usepackage{todonotes}
\usepackage[markup=underlined]{changes}
\usepackage[round]{natbib}

\newtheorem{example}{Example}
\newtheorem{definition}{Definition}
\newtheorem{remark}{Remark}

\newtheorem{theorem}{Theorem}
\newtheorem{corollary}{Corollary}
\newtheorem{proposition}{Proposition}
\newtheorem{lemma}{Lemma}

\usepackage{amsmath,amssymb,amsfonts}
\usepackage{soul}

\DeclareMathOperator*{\argmax}{arg\,max}

\newcommand{\rhot}{\tilde{\rho}}
\newcommand{\Ut}{\tilde{\mathcal{U}}}
\newcommand{\U}{\mathcal{U}}

\newcommand{\A}{\mathcal{A}}
\newcommand{\E}{\mathbb{E}}
\newcommand{\PP}{\mathcal{P}}

\title{Robust quasi-convex risk measures and applications 
\thanks{All the authors are members of Gruppo Nazionale per l’Analisi Matematica, la Probabilità e le loro Applicazioni (GNAMPA), Italy} }

\author{Francesca Centrone \thanks{Department of Economics and Business Studies, University of Eastern Piedmont, 28100 Novara, Italy. francesca.centrone@uniupo.it} 
\and Asmerilda Hitaj \thanks{Dipartimento di Economia, Università dell’Insubria, via Monte Generoso 71, 21100 Varese, Italy. asmerilda.hitaj@uninsubria.it} 
\and Elisa Mastrogiacomo \thanks{Dipartimento di Economia, Università dell’Insubria, via Monte Generoso 71, 21100 Varese, Italy. elisa.mastrogiacomo@uninsubria.it} 
\and Emanuela Rosazza Gianin \thanks{Department of Statistics and Quantitative Methods, University of Milano-Bicocca, Via Bicocca degli Arcimboldi 8, 20126 Milan, Italy. emanuela.rosazza1@unimib.it}}

\date{This version: \today}

\begin{document}

\maketitle

\begin{abstract}
This paper develops a unified framework for the robustification of risk measures beyond the classical convex and cash-additive setting. We consider general risk measures on $L^p$ spaces ($p\in [1,+\infty]$) and construct their robust counterparts through families of uncertainty sets that capture ambiguity. Two complementary mechanisms generate robust quasi-convex measures: in the first, quasi-convexity is inherited from the initial risk measure under convex uncertainty sets; in the second it comes from the quasi-convex (or c-quasi-convex) structure of the uncertainty sets themselves.  
Building on \citet{CVMMM11,frittelli-maggis}, we derive dual (penalty-type) representations for robust quasi-convex and cash-subadditive risk measures, showing that the classical convex cash-additive case arises as a special instance. We further analyze acceptance families and capital allocation rules under robustification, highlighting how ambiguity affects acceptability and the distribution of capital. 
\end{abstract}

\section{Introduction}
Quite recently, the construction of robust versions of risk measures has received increasing attention as a way to incorporate model uncertainty and ambiguity\footnote{The terms ambiguity and uncertainty are used interchangeably throughout the paper.} into financial decision-making. Various approaches to the robustness issue are possible, depending on different sources of uncertainty that can arise, just to recall some, from partial knowledge about the distributions of the risks or from the ambiguity of a Decision Maker's preferences (see, e.g., among the many, \cite{WangZiegel21}, \cite{WangXu23}). A central line of research also extends the classical framework of convex, cash-additive risk measures to robust settings. 
In particular, \citet{righi-arxiv} has recently introduced uncertainty at the level of any financial position $X$ (instead of considering a fixed uncertainty set) and has defined robust convex risk measures as: \begin{equation}
 \label{eq:rho-robust-def}
 \tilde\rho(X)\triangleq\sup_{Z\in \U_X}\rho(Z), \quad \mbox{ for any } X\in L^p,
 \end{equation}
 where $\rho$ is a cash-additive and convex risk measure -- hereafter referred to as the initial risk measure -- and $\U_X$ represents the uncertainty set for $X$ that is, a set of random variables representing possible perturbations or alternative realizations of $X$. \citet{righi-sifin} has extended the robust approach to risk sharing, also providing concrete examples based on the $p$-norm and the Wasserstein distance. \citet{pesenti-etal} have considered the dynamic framework and clarified how the choice of uncertainty sets affects time consistency. 

However, as discussed in the literature, convexity and cash-additivity may be too restrictive in realistic settings. In markets with frictions, illiquidity, or nonlinear pricing, convex combinations of positions need not preserve linear costs, and the addition of cash may not reduce risk one-to-one. This motivates the use of cash-subadditive and quasi-convex  risk measures, which generalize the convex, cash-additive paradigm while preserving economic interpretability.  
Cash-subadditivity has been introduced in \citet{ELK-R} to address discounting ambiguity due to uncertain interest rates.  Quasi-convex risk measures and their dual representation have been investigated from an axiomatic point of view in \citet{CVMMM11}, while their conditional counterpart has been studied in \citet{frittelli-maggis}. \citet{drapeau-kupper} have studied risk functionals that satisfy only quasi-convexity and monotonicity and have provided a one-to-one correspondence with risk-acceptance families and risk orders. 

Applications of quasi-convex risk measures confirm their flexibility. \citet{mastrog-rg-mor,mastrog-rg-mafe}, for instance, analyze portfolio optimization and optimal risk sharing under quasi-convex preferences, showing that such measures preserve diversification while capturing nonlinear pricing and illiquidity. Their results highlight the practical relevance of quasi-convex risk measures and motivate the need for a systematic robustification theory beyond the convex, cash-additive setting. 

Despite significant advances in representation and applications of quasi-convex and cash-subadditive risk measures, the construction of robust counterparts for such functionals under ambiguity has not yet been systematically addressed. This work develops a unified theory of robust quasi-convex (possibly cash-subadditive) risk measures and contributes to the existing literature in several ways, as summarized here below. 

First, we identify and connect two distinct approaches to construct a robust quasi-convex risk measure $\tilde{\rho}$ of the form \eqref{eq:rho-robust-def}. The first approach begins with a quasi-convex risk measure $\rho$ and a family $\U \triangleq (\U_X)_{X \in L^p}$ of uncertainty sets satisfying the condition: 
$$
\U_{\lambda X+ (1-\lambda )Y}\subseteq \lambda \U_{X}+ (1-\lambda)\U_Y, \quad \mbox{ for any } X,Y \in L^p, \lambda \in [0,1].
$$
This condition is commonly referred to as convexity of the uncertainty sets in the literature (see, e.g., \cite{righi-arxiv, pesenti-etal}).
Alternatively, we show that $\tilde{\rho}$ remains quasi-convex even without any assumption on $\rho$ (other than monotonicity), provided that the uncertainty sets satisfy: 
$$
\U_{\lambda X+ (1-\lambda )Y}\subseteq \U_{X} \cup\U_Y, \quad \mbox{ for any } X,Y \in L^p, \lambda \in [0,1].
$$ 
We refer to this property as quasi-convexity, borrowing this terminology from the literature on set-valued mappings (see, e.g., \cite{Seto-et-al2018}). Finally, we explore an alternative notion of quasi-convexity  -- defined with respect to the cone of positive financial positions  -- and demonstrate that this framework also yields a robust quasi-convex risk measure. 

Second, we characterize the largest family of uncertainty sets that generates a given
robust risk measure, extending the correspondence in \citet{pesenti-etal} beyond cash-additivity; we also identify the conditions under which robustification preserves monotonicity, (quasi-)convexity, law invariance, and continuity from above, clarifying how these properties are transferred from the (initial) risk measure $\rho$ and the family $\U$ of uncertainty sets. 

Third, building on \citet{CVMMM11,frittelli-maggis}, we derive dual representations for robust quasi-convex (not necessarily cash-subadditive)
risk measures, and show that the classical convex cash-additive case  -- including the dual forms in \citet{righi-arxiv}  -- arises as a special case of our formulation.  

Fourth, adopting the acceptance-family point of view of \citet{drapeau-kupper}, we relate the
acceptance families of $\rho$ to those of $\tilde{\rho}$ and provide tractable characterizations in a representative example. 

Finally, we extend the robustification framework to the capital allocation problem and provide illustrative examples based on widely used uncertainty sets, such as the Wasserstein and the $p$-norm balls, and quasi-convex families commonly used in set-valued analysis (see, e.g., \citet{Seto-et-al2018}).
\bigskip

The paper is organized as follows. Section~\ref{sec:preliminaries} introduces the setting and recalls the main notions of risk measures, with an emphasis on quasi-convexity, cash-subadditivity, and their dual representations. Section~\ref{sec: robust-uncertainty} develops the framework of (families of) uncertainty sets and robustification, presenting general properties and the correspondence with the largest family of uncertainty sets inducing $\tilde{\rho}$. Section~\ref{sec:dual_representation} establishes dual representations of robustified risk measures
under different assumptions on the initial risk measure and on the families of uncertainty sets, including the quasi-convex, cash-subadditive, and convex cash-additive cases. Section~\ref{sec:acceptance_families} studies
acceptance families of robust risk measures and their relationship with those of the initial risk measure. Applications to capital allocation principles and to the law invariant case with uncertainty sets based on the Wasserstein distance as well as some illustrative examples are provided in Section~\ref{sec: appl car-examples}. 
Finally, Section~\ref{sec:conclusions} concludes and outlines possible directions for future research.

\section{Preliminaries and setting}
\label{sec:preliminaries}

In this section, we introduce the notations and setting used in the paper, as well as the definitions and basic results on risk measures that will be useful later.

\subsection{Setting and notations}

We fix a probability space $(\Omega,\mathcal{F},P)$ where, as usual, $\Omega$ is a set of possible states of the world,
$\mathcal{F}$ is a $\sigma$-algebra of events, and $P$ is a reference
probability measure. Throughout the paper, random variables on $(\Omega,\mathcal{F},P)$ will represent profits and losses of financial positions where positive values indicate gains while negative ones losses. All equalities and inequalities between random variables should be meant to hold $P$-almost surely.
The expectations under the reference probability
$P$ will be denoted by $\mathbb{E}[\,\cdot\,]$. $\mathcal{P}$ stands for the set formed by all the probability measures that are absolutely continuous with respect to
$P$. The Radon–Nikodym derivative of $Q \in \mathcal{P}$ with respect to $P$ is denoted by $\tfrac{dQ}{dP}$.

We work with $L^p \triangleq L^p(\Omega,\mathcal{F},P)$ spaces,
with $p\in[1,+\infty]$, formed by all random variables that are $p$-integrable when $p \in [1,+\infty)$ or essentially bounded when $p=+\infty$. $L^p$ is endowed with the norm topology for $p \in [1,+\infty)$, while with the weak topology $\sigma(L^{\infty}, L^1)$ for $p=+\infty$. We denote $L^p_+ \triangleq \{ X \in L^p: X\geq 0\}$.

\subsection{Risk measures}

Financially speaking, we recall that risk measures can be interpreted as tools to assess the riskiness of financial positions. See \citet{ADEH-1999, Delb,follm-schied,fritt-rg} (and many subsequent papers) for a detailed treatment of the fundamentals of risk measures.

In this paper, with risk measure we mean a functional $\rho:L^p\to[-\infty,+\infty]$ that is \emph{monotone decreasing}, i.e.,
$X\le Y$ implies that $\rho(X)\ge \rho(Y)$.
Note that this convention of decreasing monotonicity is different from the ``actuarial''
convention assuming increasing monotonicity (see, e.g., \citet{pesenti-etal}).
\smallskip

The following additional axioms are sometimes imposed to risk measures:
\begin{itemize}
	\item convexity: $\rho(\lambda X+(1-\lambda)Y)\le \lambda \rho(X)+(1-\lambda)\rho(Y)$
	for all $X,Y\in L^p$, $\lambda\in[0,1]$;
	\item cash-additivity: $\rho(X+m)=\rho(X)-m$ for all $X\in L^p$, $m\in\mathbb{R}$;
    \item quasi-convexity: $\rho(\lambda X+(1-\lambda)Y)\le \max\{\rho(X),\rho(Y)\}$, for all $X, \, Y \in L^p$, $\lambda \in [0, 1]$;
	\item cash-subadditivity:
	$\rho(X+m)\le \rho(X)-m$ for all $X\in L^p$, $m\ge 0$;
    \item continuity from above: for any decreasing sequence $(X_n)_{n \in \mathbb{N}}$ in $L^p$ with $X_n \downarrow X$ as $n \to + \infty$, then $\lim_{n \to + \infty} \rho(X_n)=\rho(X)$.
\end{itemize}
See, among others, \citet{ADEH-1999, Delb,follm-schied,fritt-rg} for a discussion on coherent/convex risk measures, \citet{ELK-R} on cash-subadditive risk measures, and \citet{CVMMM11,frittelli-maggis} on quasi-convex risk measures.

It is worth recalling that convexity captures and incentivates risk diversification, while cash-additivity reflects the assumption that adding $m$ units of sure cash
reduces risk exactly by $m$. This allows to interpret cash-additive risk measures as capital requirements (or margins). 

In many applications, however, the axioms of convexity and cash-additivity are too restrictive. In markets with frictions, illiquidity or nonlinear pricing, convex
combinations of positions may not preserve linear costs, and under ambiguity on (stochastic) interest rates the effect of injecting cash is not necessarily linear. For this reason, cash-subadditive risk measures were introduced by \citet{ELK-R} to deal with discounting ambiguity, while quasi-convex risk measures were considered firstly by \citet{CVMMM11} to incentivate diversification of risk without cash-additivity.

We recall (see \citet{CVMMM11,frittelli-maggis}) that if $\rho: L^p \to [-\infty
,+\infty]$, with $p \in [1, +\infty]$, is monotone, quasi-convex, and continuous from above, then it admits the following dual representation
\begin{equation} \label{eq: repr rho}
\rho(X)= \sup_{Q \in \mathcal{P}} R_{\rho}(\E_Q[-X],Q), \quad \mbox{ for any } X \in L^p,  
\end{equation}
with a ``penalty-type'' functional 
\begin{equation} \label{eq: R}
R_{\rho}(t,Q)=R(t,Q) \triangleq \inf\{ \rho(Y): \E_Q[-Y]=t \}, \quad  \mbox{ for any } t \in \mathbb{R}, Q \in \mathcal{P},
\end{equation}
being a map $R_\rho:\mathbb{R}\times \mathcal{P}\to[-\infty,+\infty]$ that is monotone increasing and quasi-concave in $t$, with $\inf_{t \in \mathbb{R}} R(t,Q)=\inf_{t \in \mathbb{R}} R(t,Q')$ for any $Q, Q' \in \mathcal{P}$. See Theorem 2.9, Corollary 2.14, Lemma 3.2 of  \citet{frittelli-maggis}, and Theorem 3.1 of \citet{CVMMM11}.

In the special case of convex cash-additive risk measures $\rho$, $R_{\rho}$ reduces to
\[
R_\rho(t,Q)=t-c_\rho(Q), \quad  \mbox{ for any } t \in \mathbb{R}, Q \in \mathcal{P},
\]
where $c_\rho(Q)$ is the (minimal) penalty functional in the dual representation of $\rho$, i.e. 
\begin{equation} \label{eq: c-rho}
c_{\rho}(Q) \triangleq \sup_{X \in L^p} \{\E_Q[-X]-\rho(X)\}, \quad \mbox{ for any } Q \in \mathcal{P}.
\end{equation}
See \cite{fritt-rg, follm-schied, frittelli-maggis}. 

Quite recently, in order to model ambiguity in the risk evaluation \citet{righi-arxiv, righi-sifin} and \citet{pesenti-etal} built \textit{robust risk measures} by means of uncertainty sets. For each $X\in L^p$, an uncertainty set $\U_X \subseteq L^p$ is a set of random variables representing possible perturbations or alternative realizations of $X$. 

We recall from the aforementioned papers that, once the initial cash-additive risk measure $\rho$ is fixed,
\begin{itemize}
    \item given a family of uncertainty sets $\U=(\U_X)_{X \in L^{p}}$ (satisfying suitable properties), the induced $\U$-robust risk measure $\rhot$ is defined as 
   \begin{equation}\label{eq: rho robust}
\rhot(X)\triangleq \rhot^{\U}(X) \triangleq\sup_{Z \in \U_X} \rho(Z), \quad \mbox{ for any } X \in L^{p}.   
\end{equation}
In other words, $\tilde\rho(X)$ represents the worst-case value of $\rho$ when $Z$ may vary within the uncertainty set $\U_X$.
\end{itemize}

Furthermore, a one-to-one correspondence between robust risk measures and families of uncertainty sets can be found in \citet{pesenti-etal} in terms of the largest family of uncertainty sets. Indeed,  
\begin{itemize}
\item given a robust cash-additive risk measure $\rhot$, the largest family of uncertainty sets inducing $\rhot$ (called ``consolidated'' in \citet{pesenti-etal}) is defined as $\Ut \triangleq (\Ut_X)_{X \in L^{p}}$, with
\begin{equation} \label{eq: largest U}
\Ut_X \triangleq \bigcup \big\{\U_X \mbox{ uncertainty set}: \rhot(X)= \rhot^{\U}(X) \big\}, \quad \mbox{ for any } X \in L^p.
\end{equation}
\end{itemize}
We address to \citet{righi-arxiv,pesenti-etal} for precise statements and assumptions.






\section{Uncertainty sets and robustification}\label{sec: robust-uncertainty}

A central component in robustifying a risk measure is the choice of the structural properties of $\mathcal{U}$ -- such as monotonicity, convexity, or (quasi-)convexity -- which directly influence the regularity of the robustified functional. Two complementary approaches can be adopted to obtain a robust quasi-convex risk measure. In the first, quasi-convexity is inherited from the risk measure $\rho$, under convexity of the family  $\mathcal{U}$ of the uncertainty sets; in the second, the initial functional $\rho$ is left general (i.e., quasi-convexity is not needed), and quasi-convexity (or the stronger property of c-quasi-convexity) is instead imposed on the uncertainty sets of the family $\mathcal{U}$. In this case, the source of quasi-convexity shifts from the initial functional to the geometry of the uncertainty sets, so that even a merely monotone $\rho$ yields a robust quasi-convex risk measure. In the following, we formalize these two approaches.

We begin by introducing the structural assumptions on the families $\U$ of uncertainty sets that will be sometimes required later. Later on, we will study the impact of such assumptions on the corresponding robust risk measure. When not specified otherwise, the elements $(\U_X)$ of the family $\U$ are meant to belong to $L^p$ with general $p \in [1,+\infty]$.\bigskip

\begin{definition}
A family $\U \triangleq (\U_X)_{X \in L^{p}}$ of uncertainty sets is said to satisfy:    

\noindent \textit{- Monotonicity:} if $X \leq Y$ then $\U_X \supseteq \U_Y$.

\noindent \textit{- Order preservation:} if $X \leq Y$ then for any $Y^\prime \in \U_Y$ there exists  $X^\prime \in \mathcal{U}_X$ such that $X^\prime \leq Y^\prime$.

\noindent \textit{- Convexity: } $\U_{\lambda X+ (1-\lambda )Y}\subseteq \lambda \U_{X}+ (1-\lambda)\U_Y$ for any $X,Y \in L^{p}$, $\lambda \in [0,1]$.

\noindent \textit{- Quasi-convexity: } $\U_{\lambda X+ (1-\lambda )Y}\subseteq \U_{X} \cup\U_Y$ for any $X,Y \in L^{p}$, $\lambda \in [0,1]$.

\noindent \textit{- c-quasi-convexity: } for any $X,Y \in L^{p}$, $\lambda \in [0,1]$
\begin{equation}\label{eq:quasi-co}
    \U_{\lambda X+ (1-\lambda )Y}\subseteq \U_{X} \cup \U_Y +L^p_+.
\end{equation}

\noindent \textit{- Continuity from above:} if $(X_n)_{n \in \mathbb{N}}$ is a decreasing sequence in $L^p$ with $X_n \downarrow X$ as $n \to + \infty$, then  $\cup_{n \in \mathbb{N}}\U_{X_n}=\U_X$. 

\noindent \textit{- Solidity:} for any $X \in L^{p}$ it holds that: if $Y \in \U_X$ and $\bar{Y} \geq Y$ (with $\bar{Y} \in L^p$), then also $\bar{Y} \in \U_X$.

\noindent \textit{- Law invariance:} if $X \sim X'$, i.e. $X$ and $X'$ have the same distribution with respect to $P$, then $\U_X= \U_{X'}$.

\noindent {- Cash-invariance:} for any $X \in L^{p}$, $c \in \mathbb{R}$, it holds that: $\U_{X+c}=\U_X +c.$

\end{definition}

For what concerns the above properties of families of uncertainty sets, convexity is the same as in \citet{righi-arxiv}, while monotonicity and order preservation correspond to the analogous axioms in \citet{pesenti-etal}. However, we observe that order preservation is called monotonicity in \citet{righi-arxiv}. Quasi-convexity of the family of uncertainty sets is the ``natural'' generalization of the classical quasi-convexity notion to set-valued maps (where, in our context, we consider the set-valued map which associates to each $X$ an uncertainty set $\U_X$). In this framework, the order relation $\leq$ is replaced by the inclusion operator and supremum corresponds to union of sets. As for c-quasi-convexity, it can be found in the literature on set-valued maps (see the notion of (u3)-type c-quasi-convexity in \citet{Seto-et-al2018} with the cone $C$ being $-L^{p}_+$ in our context). 

\begin{remark}\label{rem:qco-cvx-relation}
    (a) \ We observe that c-quasi-convexity (defined as in \eqref{eq:quasi-co}) is not implied in general by the convexity assumption of \citet{righi-arxiv}. See Example \ref{ex: conv-not qco} below and \citet[Section 4, pag.13]{Seto-et-al2018} for further comments.

    (b) \ Several other types of convexity and quasi-convexity for families of sets could be also introduced (see, e.g., \citet[pp.5--6]{Seto-et-al2018}), and the properties of the corresponding robust risk measures could then be investigated.
\end{remark}

The following example illustrates that in general convexity does imply neither quasi-convexity nor c-quasi-convexity. Examples of (c-)quasi-convex families of uncertainty sets will be provided in Section \ref{sec: examples}.

\begin{example}\label{ex: conv-not qco}
Fix $\varepsilon >0$ and consider $\U$ on $L^{\infty}$ with
\begin{equation*}
\U_X= X+ B_{\varepsilon}, \quad \mbox{ for any } X \in L^{\infty}, 
\end{equation*}
where $B_{\varepsilon}\triangleq \{Z \in L^{\infty}: \Vert Z\Vert_{\infty} \leq \varepsilon \}$.

It follows immediately that the family $\U$ is convex. 
Nevertheless, $\U$ is neither quasi-convex nor c-quasi-convex. 

To show this, let  $X=0$, $Y=10 \varepsilon$, and $\alpha =\frac12$, it holds  $\U_{\alpha X + (1-\alpha)Y}=5 \varepsilon + B_{\varepsilon}$. For  $\bar{Z}=\frac12 \varepsilon$, then, it follows that
 $5 \varepsilon + \bar{Z}= \frac{11}{2} \varepsilon \in \U_{\alpha X + (1-\alpha)Y}$ but it does not belong to $\U_X \cup \U_Y=B_{\varepsilon} \cup (10\varepsilon+B_{\varepsilon})$. This implies that $\U_{\lambda X+ (1-\lambda )Y}\nsubseteq \U_{X} \cup\U_Y$, that is, quasi-convexity of $\U$ fails.
\smallskip

With a slight modification of this example, it can be also shown that convexity of $\U$ does not guarantee c-quasi-convexity.
\end{example}

The following Lemma clarifies the link between the two concepts of quasi-convexity and c-quasi-convexity for families of uncertainty sets.

\begin{lemma} \label{lem: c-qco}
i) Solidity of $\U$ is equivalent to $\U_X=\U_X + L^{p}_+$ for any $X \in L^{p}$.

\noindent ii) If the family $\U$ 
is quasi-convex, then it is also c-quasi-convex.
The converse implication holds under solidity of $\U$.
\end{lemma}

\begin{proof}
i) If $\U_X=\U_X + L^{p}_+$ holds for any $X \in L^{p}$, then solidity is straightforward. Conversely, if solidity holds, then the inclusion $\subseteq$ is immediate. Concerning the $\supseteq$ inclusion, if $Z \in \U_X + L^{p}_+$ then $Z \geq \tilde{Z}$ for some $\tilde{Z} \in \U_X$. By solidity of $\U$, it follows that also $Z \in \U_X$.

\noindent ii) The first statement is straightforward. 
The converse implication is due to
\begin{equation*}
\U_X\cup \U_Y+L_+^{p}=(\U_X+L_+^{p})\cup(\U_Y+L_+^{p})= \U_X\cup \U_Y,
\end{equation*}
where the former equality can be checked easily, the latter follows from solidity.
\end{proof}
\bigskip

In the sequel, we will combine the approaches of \citet{righi-arxiv} and \citet{pesenti-etal} (recalled at the end of Section \ref{sec:preliminaries}, equations \eqref{eq: rho robust} and \eqref{eq: largest U}) to robustify a risk measure $\rho$ that is not necessarily cash-additive. 
Also, in the present setting where we consider $L^p$ spaces and we do not necessarily  require cash-additivity, it can be proved (as in \citet{pesenti-etal}, Lemma 3) that the largest family of uncertainty sets inducing $\rhot$ corresponds to:
\begin{equation} \label{eq: largest uncertainty}
\Ut_X= \big\{Z \in L^{\infty}: \rho(Z) \leq \rhot(X) \big\}, \quad \mbox{for any } X \in L^{p}. 
\end{equation}
\smallskip

In the following result, items a) and b) generalize items 2. and 7. of Theorem 2 of \citet{pesenti-etal} to the case of robustification of risk measures $\rho$ that are not necessarily cash-additive, for general families of uncertainty sets  in $L^p$ spaces, while items c), d) and e) are new and deal with quasi-convexity of the robust risk measure. 

\begin{proposition} \label{prop: properties robust-cqco}
Let $\rho$ be a risk measure, $\U$ be a family of uncertainty sets and $\rhot$ the robustification of $\rho$ via $\U$.

\noindent a) If $\U$ is monotone or order preserving then $\rhot$ is monotone.

\noindent b) If $\rho$ is convex and $\U$ is a convex family, then $\rhot$ is convex. 

\noindent c) If $\U$ is c-quasi-convex (or quasi-convex), then $\rhot$ is quasi-convex.

\noindent d) If $\rho$ is quasi-convex and $\U$ is convex, then $\rhot$ is quasi-convex.

\noindent  e) If $\U$ is monotone and continuous from above, then $\rhot$ is continuous from above.

\noindent  f) If $\U$ is law invariant, then $\rhot$ is law invariant.
\end{proposition}

\begin{proof}
a) Let $X,Y \in L^p$ with $X \leq Y$ be fixed arbitrarily.

If $\U$ is monotone, then $\U_X \supseteq \U_Y$, implying that
$$
\rhot(X)= \sup_{Z \in \U_X} \rho(Z) \geq \sup_{Z \in \U_Y} \rho(Z)  = \rhot(Y),
$$
that is, monotonicity of $\rhot$. 

If $\mathcal{U}$ is order preserving, then for any $Y^\prime \in \mathcal{U}_Y$ there exists $X^\prime \in \mathcal{U}_X$, such that $X^\prime \leq Y^\prime$. 
Then, by monotonicity of $\rho$, we have $\rho(X^\prime)\geq \rho(Y^\prime)$. 
Taking into account that such an inequality holds for all $Y^\prime \in \mathcal{U}_Y$ and by definition of $\rhot$, we obtain 
$\tilde{\rho}(X) \geq \tilde{\rho}(Y)$. 

\noindent b) Suppose that $\rho$ is a convex risk measure and that $\U$ is convex.
By definition of $\rhot$, it holds that for any $X,Y \in L^{p}$ and $\alpha \in [0,1]$
\begin{align}
\rhot(\alpha X + (1-\alpha)Y) &= \sup_{Z \in \U_{\alpha X + (1-\alpha)Y}} \rho(Z) \notag\\
&\leq \sup_{Z \in (\alpha \U_{X}+(1-\alpha) \U_Y)} \rho(Z) \label{eq: ineq1}\\
&\leq \sup_{Z_1 \in \U_X; Z_2 \in \U_Y} \rho(\alpha Z_1 + (1-\alpha)Z_2) \notag\\
&\leq \sup_{Z_1 \in \U_X; Z_2 \in \U_Y} \big(\alpha \rho(Z_1) + (1-\alpha) \rho(Z_2) \big) \label{eq: ineq2}\\
&= \alpha \sup_{Z_1 \in \U_X} \rho(Z_1) + (1-\alpha) \sup_{Z_2 \in \U_Y}\rho(Z_2) \notag\\
&=\alpha  \rhot(X) + (1-\alpha) \rhot(Y), \notag
\end{align}
where \eqref{eq: ineq1} is due to convexity of $\U$ and \eqref{eq: ineq2} to convexity of $\rho$.

\noindent c) By Lemma \ref{lem: c-qco}- ii), quasi-convexity of $\U$ implies c-quasi-convexity of $\U$. Therefore, it is enough to prove that c-quasi-convexity of $\U$ implies quasi-convexity of $\rhot$. 

From c-quasi-convexity of $\U$, it follows that, for any $X,Y \in L^{p}$ and $\alpha \in [0,1]$,
 \begin{equation*}
 \rhot(\alpha X + (1-\alpha)Y)= \sup_{Z \in \U_{\alpha X + (1-\alpha)Y}} \rho(Z) \leq \sup_{Z \in \U_{X}\cup \U_Y +L^p_+} \rho(Z) .
 \end{equation*}
 If $Z \in \U_{X}\cup \U_Y +L^p_+$ then $Z= Z_{X,Y}+K$ for some $Z_{X,Y} \in \U_{X}\cup \U_Y$ and $ K \in L^p_+$. Then, by monotonicity of $\rho$,
 $$
 \rho(Z) = \rho ( Z_{X,Y}+K) \leq \rho(Z_{X,Y}).
 $$
 The previous arguments imply that
 \begin{equation*}
  \rhot(\alpha X + (1-\alpha)Y) \leq \sup_{Z \in \U_{X}\cup \U_Y +L^p_+} \rho(Z) \leq \sup_{Z \in \U_{X}\cup \U_Y } \rho(Z)  = \max\{ \rhot(X); \rhot(Y) \},
 \end{equation*}
 i.e. quasi-convexity of $\rhot$.

\noindent d) Proceeding similarly as in item b), convexity of $\U$ implies that, for any $X,Y \in L^{p}$ and $\alpha \in [0,1]$,
\begin{align}
\rhot(\alpha X + (1-\alpha)Y) &= \sup_{Z \in \U_{\alpha X + (1-\alpha)Y}} \rho(Z) \notag\\
&\leq \sup_{Z_1 \in \U_X; Z_2 \in \U_Y} \rho(\alpha Z_1 + (1-\alpha)Z_2) \notag\\
&\leq \sup_{Z_1 \in \U_X; Z_2 \in \U_Y} \max\{\rho(Z_1); \rho(Z_2) \} \label{eq: ineq2-b}\\
&= \max\bigg\{\sup_{Z_1 \in \U_X} \rho(Z_1); \sup_{Z_2 \in \U_Y}\rho(Z_2) \bigg\} \notag\\
&= \max\{\rhot(X); \rhot(Y) \}, \notag
\end{align}
where \eqref{eq: ineq2-b} is due to quasi-convexity of $\rho$.

\noindent e) 
Let $(X_n)_{n \in \mathbb{N}}$ be a decreasing sequence of $L^p$ with $X_n \downarrow X$ as $n \to + \infty$. Then
\begin{align}
\rhot(X) &\geq \lim_{n \to + \infty} \rhot(X_n) = \sup_{n \in \mathbb{N}} \rhot(X_n) \label{eq: 999}\\
&=\sup_{n \in \mathbb{N}}\sup_{Z \in \U_{X_n}} \rho(Z) \notag\\
&= \sup_{Z \in \U_{X}} \rho(Z) =\rhot(X), \label{eq: 999-1}
\end{align}
where \eqref{eq: 999} is due to monotonicity of $\U$ --hence, by item a), monotonicity of $\rhot$--, while the first equality in \eqref{eq: 999-1} is implied by continuity from above of $\U$.
Continuity from above of $\rhot$ then follows.

\noindent f) is immediate.
\end{proof}

In other words, Proposition \ref{prop: properties robust-cqco} underlines that robust risk measures can be built at least in two ways: with a general risk measure $\rho$ and a monotone quasi-convex (or c-quasi-convex) family $\U$ of uncertainty sets; or with a quasi-convex risk measure $\rho$ and a monotone convex $\U$. Notice that the two approaches are different since, as already observed in Remark \ref{rem:qco-cvx-relation}, convexity of $\U$ does imply neither quasi-convexity nor c-quasi-convexity. Roughly speaking, the ``sources'' of quasi-convexity of $\rhot$ can be either at the level of the primal risk measure $\rho$ or of the family $\U$.
\medskip

Consider now the largest family $\Ut$ of uncertainty sets inducing a given $\rhot$ as in \eqref{eq: largest U}. Several properties of $\Ut$ have been established in \citet{pesenti-etal} for cash-additive risk measures.

Item a) of the following result generalizes item 2. of Theorem 2 of \citet{pesenti-etal} to the case of robustification of risk measures $\rho$ that are not necessarily cash-additive, while item b) is new.

\begin{proposition} 
Let $\rho$ be a risk measure, $\rhot$ the robustification of $\rho$ and $\Ut$ the largest family of uncertainty sets defined in \eqref{eq: largest U} associated to $\rhot$.

Then $\Ut$ is solid and the following implications hold.

\noindent a) If $\rhot$ is monotone, then $\Ut$ is monotone.

\noindent b) If $\rhot$ is quasi-convex, then $\Ut$ is quasi-convex (hence, also c-quasi-convex).
\end{proposition}

\begin{proof}
\noindent Solidity of $\Ut$ follows immediately from \eqref{eq: largest uncertainty} and from the monotonicity of $\rho$.

\noindent a) can be proved similarly to Theorem 2, item 2., of \citet{pesenti-etal}.
Assume, indeed, that $\rhot$ is monotone. If $X \leq Y$, then $\rhot(Y) \leq \rhot(X)$, hence
$$
\Ut_Y=\{ Z \in L^{p}: \rho(Z) \leq \rhot(Y)\} \subseteq \{ Z \in L^{p}: \rho(Z) \leq \rhot(X)\}=\Ut_X.
$$

\noindent b) Assume that $\rhot$ is quasi-convex.
 Let $X,Y \in L^{p}$ and $\alpha \in [0,1]$ be fixed arbitrarily. If $Z \in \Ut_{\alpha X + (1-\alpha)Y}$, then 
 $$
 \rho(Z) \leq \rhot(\alpha X + (1-\alpha)Y) \leq \max\{\rhot(X); \rhot(Y)\},
 $$
 where the former inequality comes from \eqref{eq: largest uncertainty}, while the latter from quasi-convexity of $\rhot$. Again by \eqref{eq: largest uncertainty}, it follows that $Z \in \Ut_X \cup \Ut_Y$, so $\Ut$ is quasi-convex.
\end{proof}

\section{Dual representation of robust quasi-convex risk measures}
\label{sec:dual_representation}

In this section, we will provide the dual representation of robust risk measures in the two approaches discussed in Section \ref{sec: robust-uncertainty}, both building robust quasi-convex risk measures.
As in \citet{righi-arxiv}, a key ingredient of the dual representation is the support function of an uncertainty set, defined as
\begin{equation} \label{eq: varphi_Q}
\varphi_Q(X) \triangleq \sup_{Z \in \U_X} \E_Q[-Z], \quad \mbox{ for any } X \in L^{p},
\end{equation}
once a family $\U$ of uncertainty sets is fixed.
\smallskip

\textit{In the sequel, families $\U$ of uncertainty sets will be always supposed to satisfy monotonicity.}

\subsection{First approach: quasi-convex $\rho$, convex $\U$}

We will now prove a dual representation of robust quasi-convex risk measures in the first approach of robustification, that is, for a quasi-convex $\rho$ and a convex family $\U$ of uncertainty sets. 
\smallskip

\begin{proposition} 
\label{prop: dual-repr-gen}
If $\U$ is a convex family of uncertainty sets and $\rho: L^{p} \to [-\infty, + \infty]$, with $p \in [1,+\infty]$, is a quasi-convex and continuous from above risk measure, then the associated robust risk measure $\rhot$ is quasi-convex.
Furthermore, it has the following dual representation:
\begin{equation*} 
	\rhot(X) = \sup_{Q \in \PP} \sup_{Z \in \U_X} R_{\rho} \left( \E_Q[-Z], Q \right), \quad \mbox{ for any } X \in L^{p},
\end{equation*}
where $R_{\rho}$ is the penalty-type functional in \eqref{eq: R}  of $\rho$.
\end{proposition}

\begin{proof}
Monotonicity and quasi-convexity of $\rhot$ follow, respectively, from items a) and d) of Proposition \ref{prop: properties robust-cqco}.

From Theorem 2.9, Corollary 2.14, Lemma 3.2 of \citet{frittelli-maggis} (see also Propositions 4.3-4.4 and Theorem 3.1 of \citet{CVMMM11} on $L^{\infty}$), $\rho$ can be represented as 
\begin{equation*} 
\rho(X)= \sup_{Q \in \PP} R_\rho \left( \E_Q[-X], Q \right), \quad \mbox{ for any } X \in L^{p},
\end{equation*}
where $R_{\rho}(t,Q)$ is
monotone increasing and quasi-concave in $t \in \mathbb{R}$, with $\inf_{t \in \mathbb{R}} R(t,Q)=\inf_{t \in \mathbb{R}} R(t,Q')$ for any $Q, Q' \in \mathcal{P}$.

By the previous arguments and by definition of $\rhot$, it follows that, for any $X \in L^p$,
\begin{align*}
\rhot(X)&=\sup_{Z \in \U_X} \rho(Z) =\sup_{Z \in \U_X} \sup_{Q \in \PP} R_\rho \left( \E_Q[-Z], Q \right) \\
&=\sup_{Q \in \PP} \sup_{Z \in \U_X}  R_\rho \left( \E_Q[-Z], Q \right).
\end{align*}
\end{proof}

Under stronger assumptions on $\rho$, we will provide a ``more explicit'' dual representation of the robust quasi-convex risk measure. 
With this in mind, we will start with the particular case of risk measures $\rho$ corresponding to certainty equivalents on $L^{\infty}$, i.e.
\begin{equation*} 
     \rho_\ell(X) = \ell^{-1} \left( \E_P\left[\ell(-X)\right] \right), \quad \mbox{ for any } X \in L^{\infty},
\end{equation*}
where $\ell: \mathbb{R} \to \mathbb{R}$ is a strictly increasing convex function (also called loss function) and $\ell^{-1}$ denotes its inverse function.
By Proposition 5.3 in \citet{CVMMM11}, $\rho_\ell$ is quasi-convex and admits the following robust representation:	
\begin{equation*} 
    \rho_\ell(X) = \sup_{Q \in \PP} \, R_\ell \left( \E_Q[-X], Q \right), \quad \mbox{ for any } X \in L^{\infty},
\end{equation*}
where $R_\ell(t, Q)$ is a penalty-type function of the following form:
\begin{equation} \label{eq: penalty_function}
	R_\ell(t, Q) \triangleq \ell^{-1} \left( \max_{x \geq 0} \left\{ x t - \E_P \left[ \ell^* \left( x \, \frac{dQ}{dP} \right) \right] \right\} \right), \quad \mbox{ for any } (t,Q) \in (\mathbb{R},\mathcal{P}),
\end{equation}
with $\ell^*$ denoting the convex conjugate of $\ell$.
\smallskip

Now, consider a robustified version of the risk measure $\rho_\ell$, where uncertainty affects the underlying outcomes. Namely,
\begin{equation} \label{eq: robustified_rho}
	\rhot_\ell(X) \triangleq \sup_{Z \in \U_X} \rho_\ell(Z)= \sup_{Z \in \U_X} \sup_{Q \in \PP} \, R_\ell \left( \E_Q[-Z], Q \right), \quad \mbox{ for any } X \in L^{\infty}.
\end{equation}

\begin{proposition}[Robustified certainty equivalent] 
\label{prop: robust cce represent}
Let $\U$ be a convex family of uncertainty sets on $L^{\infty}$. Then the  robustified certainty equivalent $\rhot_\ell$ is quasi-convex and has the following dual representation:
\begin{equation*} 
	\rhot_\ell(X) = \sup_{Q \in \PP} R_\ell \left( \varphi_Q(X), Q \right), \quad \mbox{ for any } X \in L^{\infty}.
\end{equation*}
\end{proposition}

\begin{proof}
Monotonicity  and quasi-convexity of $\rhot_\ell$ follow directly from quasi-convexity of $\rho_\ell$ and from Proposition \ref{prop: properties robust-cqco}, items a) and d).

We are now going to prove the dual representation. By interchanging the suprema in \eqref{eq: robustified_rho} and applying the explicit formulation of $R_\ell$ in \eqref{eq: penalty_function}, we obtain that, for any $X \in L^{\infty}$,
\begin{align}
	\rhot_\ell(X) &=  \sup_{Q \in \PP} \sup_{Z \in \U_X} R_\ell \left( \E_Q[-Z], Q \right) \notag\\
   &= \sup_{Q \in \PP} \sup_{Z \in \U_X} \ell^{-1} \left( \sup_{x \geq 0} \left\{ x \, \E_Q[-Z] - \E_P \left[ \ell^* \left( x \, \frac{dQ}{dP} \right) \right] \right\} \right) \notag\\
    &= \sup_{Q \in \PP} \ell^{-1} \left( \sup_{x \geq 0} \left\{ x \cdot \sup_{Z \in \U_X} \E_Q[-Z] - \E_P \left[ \ell^* \left( x \, \frac{dQ}{dP} \right) \right] \right\} \right) \label{eq: eq-cont}\\
    &= \sup_{Q \in \PP} \ell^{-1} \left( \sup_{x \geq 0} \left\{ x \cdot \varphi_Q(X) - \E_P \left[ \ell^* \left( x \, \frac{dQ}{dP} \right) \right] \right\} \right) \notag \\
    &= \sup_{Q \in \PP} R_\ell \left( \varphi_Q(X), Q \right), \label{eq: robust cce thesis}
\end{align}
where equality \eqref{eq: eq-cont} is due to continuity of the function $\ell^{-1}$ (as a consequence of its concavity on $\mathbb{R}$), while \eqref{eq: robust cce thesis} follows from \eqref{eq: penalty_function} and \eqref{eq: varphi_Q}.
\end{proof}

In the same spirit of Proposition \ref{prop: robust cce represent}, 
the following result provides a dual representation of robust risk measures for a quasi-convex cash-subadditive $\rho$ and a convex family $\U$ in general spaces $L^p$, with $p \in [1,+ \infty]$. 

\begin{theorem} \label{thm: dual represen-qco rho-Lp}
If $\U$ is a convex family of uncertainty sets on $L^p$ and $\rho: L^p \to [-\infty, + \infty]$ is a quasi-convex, cash-subadditive and continuous from above risk measure, then the associated robust risk measure $\rhot$ is quasi-convex and has the following dual representation:
\begin{equation*} 
 	\rhot(X) = \sup_{Q \in \PP} R_{\rho} \left( \varphi_Q(X), Q \right), \quad \mbox{ for any } X \in L^p.
\end{equation*}
\end{theorem}

\begin{proof}
Monotonicity and quasi-convexity of $\rhot$ follow from Proposition \ref{prop: properties robust-cqco}, items a) and d).

As in the proof of Proposition \ref{prop: dual-repr-gen}, $\rho$ can be represented as 
\begin{equation} \label{eq: dual rho-Lp}
\rho(X)= \sup_{Q \in \PP} R_\rho \left( \E_Q[-X], Q \right), \quad \mbox{ for any } X \in L^p,
\end{equation}
with $R_{\rho}$ being
monotone increasing and quasi-concave in $t$, and satisfying $\inf_{t \in \mathbb{R}} R_{\rho}(t,Q)=\inf_{t \in \mathbb{R}} R_{\rho}(t,Q')$ for any $Q, Q' \in \mathcal{P}$.

Moreover, cash-subadditivity of $\rho$ implies that $R_{\rho}$ is also non-expansive, i.e., $R_{\rho}(t',Q) \leq R_{\rho}(t,Q) + |t-t'|$ for any $t,t' \in \mathbb{R}$, $Q \in \PP$. This can be proved (on $L^p$) exactly as in the proof of Theorem 3.1 of \citet{CVMMM11}.

From the non-expansivity of $R_\rho$ it follows that $R_{\rho}(\cdot, Q)$ is left-continuous. Indeed, for any arbitrary $t \in \mathbb{R}$ and sequence $(t_n)_{n \in \mathbb{N}}$ with $t_n \uparrow t$ as $n \to + \infty$, increasing monotonicity and non-expansivity of $R_\rho(\cdot,Q)$ guarantee that
$$
0 \leq R_{\rho}(t,Q) - R_{\rho}(t_n,Q) \leq |t-t_n|.
$$
It then follows that $\lim_{n \to + \infty} R_\rho(t_n,Q)=R_\rho(t,Q)$ for any $Q \in \PP$, hence the left-continuity of $R_{\rho}(\cdot, Q)$ holds.

By \eqref{eq: dual rho-Lp}, $\rhot$ can be rewritten, for any $X \in L^{p}$, as:
\begin{align}
\rhot(X)&=\sup_{Z \in \U_X} \rho(Z) =\sup_{Z \in \U_X} \sup_{Q \in \PP} R_\rho \left( \E_Q[-Z], Q \right) \notag \\ %
&= \sup_{Q \in \PP} \sup_{Z \in \U_X} R_\rho \left( \E_Q[-Z], Q \right) \notag \\
&= \sup_{Q \in \PP}  R_\rho \left(\sup_{Z \in \U_X} \E_Q[-Z], Q \right) \label{eq: left-cont-Lp} \\
&= \sup_{Q \in \PP}  R_\rho \left(\varphi_Q(X), Q \right), \notag
\end{align}
where equality in \eqref{eq: left-cont-Lp} is due to increasing monotonicity and left-continuity of $R_{\rho}(\cdot, Q)$.
\end{proof}

Here below, we show that, when dealing with convex and cash-additive risk measures and with cash-invariant uncertainty sets, the dual representation provided in the previous result reduces to Theorem 1 of \citet{righi-arxiv}.

\begin{corollary} 
\label{cor: dual corollary-Lp}
If $\U$ is a convex, cash-invariant and continuous from above family of uncertainty sets, and $\rho: L^p \to [-\infty,+\infty]$ is a convex, cash-additive and continuous from above risk measure, then the associated $\rhot$ is convex and cash-additive and has the following dual representation:
\begin{equation} \label{eq: dual repres robut convex-Lp}
	\rhot(X) = \sup_{\widetilde{Q} \in \PP} \Big\{\E_{\widetilde{Q}}[-X] - \inf_{Q \in \PP} \{ c_{\varphi_Q}(\widetilde{Q}) +c
_{\rho}(Q)\} \Big\}, \quad \mbox{ for any } X \in L^p, 
\end{equation}
where $c_\rho$ and $c_{\varphi_Q}$ denote the minimal penalty function of $\rho$ and $\varphi_Q$, respectively.
\end{corollary}

\begin{proof}
It can be easily checked that $\rhot$ is a convex cash-additive risk measure.  

It remains to prove \eqref{eq: dual repres robut convex-Lp}.
It is well-known (see \citet{CVMMM11} and \citet{frittelli-maggis}) that if $\rho$ is a convex cash-additive risk measure, then 
\begin{equation} \label{eq: K-convex}
R_\rho(t,Q)=t-c_{\rho}(Q), \quad \mbox{ for any } (t,Q) \in (\mathbb{R},\PP).
\end{equation}
By Theorem \ref{thm: dual represen-qco rho-Lp} and by \eqref{eq: K-convex}, $\rhot$ becomes
\begin{equation} 
\label{eq: rhot-000}
	\rhot(X) = \sup_{Q \in \PP} \{ \varphi_Q(X) - c_{\rho}(Q)\}, \quad \mbox{ for any } X \in L^{p}.
\end{equation}

We prove now that $\varphi_Q$ is a convex cash-additive and continuous from above risk measure.

\noindent \textit{Monotonicity.} If $X \leq Y$, then $\U_X \supseteq \U_Y$ (by monotonicity of $\U$). It is straightforward to check that $\varphi_Q(X) \geq \varphi_Q(Y)$, i.e. decreasing monotonicity of $\varphi_Q$.

\noindent \textit{Continuity from above.} Let $(X_n)_{n \in \mathbb{N}}$ be a decreasing sequence with $X_n \downarrow X$ as $n \to + \infty$. Then
\begin{align*}
\varphi_Q(X) \geq \sup_{n \in \mathbb{N}}\varphi_Q(X_n) &=\sup_{n \in \mathbb{N}}\sup_{Z \in \U_{X_n}} \E_Q[-Z] \notag\\
&=\sup_{Z \in \cup_n\U_{X_n}} \E_Q[-Z] \notag \\
& = \sup_{Z \in \U_{X}} \E_Q[-Z]= \varphi_Q(X), 
\end{align*}
where the first equality in the last line is due to continuity from above of $\mathcal{U}$. Continuity from above of $\varphi_Q$ then follows.

\noindent \textit{Convexity} and \textit{cash-additivity} of $\varphi_Q$ follow from Lemma 1 of \citet{righi-arxiv}.\smallskip

By the previous arguments, $\varphi_Q$ can be represented as
\begin{equation}
\label{eq: rhot-0002}
\varphi_Q(X)= \sup_{\widetilde{Q} \in \PP} \big\{\E_{\widetilde{Q}}[-X] - c_{\varphi_Q}(\widetilde{Q}) \big\}, \quad \mbox{ for any } X \in L^{p}.
\end{equation}
From \eqref{eq: rhot-000} and \eqref{eq: rhot-0002}, it follows that, for any $X \in L^{p}$, 
\begin{align*}
\rhot(X) &= \sup_{Q \in \PP} \{ \varphi_Q(X) - c_{\rho}(Q)\} \\
&= \sup_{Q \in \PP} \bigg\{ \sup_{\widetilde{Q} \in \PP} \big\{\E_{\widetilde{Q}}[-X] - c_{\varphi_Q}(\widetilde{Q}) \big\} - c_{\rho}(Q) \bigg\} \\
&= \sup_{Q \in \PP} \sup_{\widetilde{Q} \in \PP} \big\{\E_{\widetilde{Q}}[-X] - c_{\varphi_Q}(\widetilde{Q})  - c_{\rho}(Q) \big\} \\
&=  \sup_{\widetilde{Q} \in \PP} \Big\{\E_{\widetilde{Q}}[-X] - \inf_{Q \in \PP} \{c_{\varphi_Q}(\widetilde{Q})  + c_{\rho}(Q)\} \Big\}.
\end{align*}
\end{proof}

\subsection{Second approach: convex $\rho$ , quasi-convex $\U$}

In this section, we provide a dual representation of robust quasi-convex risk measures obtained with the second approach, that is, for a convex $\rho$ and a quasi-convex $\U$. 

\begin{proposition} 
Let $\rho: L^p \to[-\infty, + \infty]$ be a convex, continuous from above, cash-additive risk measure and let $\U$ be a quasi-convex and continuous from above family of uncertainty sets.

Then:

\noindent a) for any $Q \in \mathcal{P}$, $\varphi_Q$ satisfies quasi-convexity, decreasing monotonicity and continuity from above;

\noindent b) $\rhot$ is quasi-convex and has the following dual representation:
\begin{equation*} 
\rhot(X)=\sup_{Q, \tilde{Q} \in \mathcal{P}}  \big\{  R_{\varphi_Q}(\E_{\tilde{Q}}[-X],\tilde{Q}) - c_{\rho}(Q) \big\}, \quad \mbox{ for any } X \in L^p,
\end{equation*}
where $c_{\rho}$ is the minimal penalty function of $\rho$ and $R_{\varphi_Q}$ is the penalty-type functional of $\varphi_Q$.
\end{proposition}

\begin{proof}
a) \textit{Quasi-convexity.}  By quasi-convexity of $\U$, for any $X,Y \in L^p$ and $\alpha \in [0,1]$ it holds that
 \begin{align*}
\varphi_Q(\alpha X+ (1-\alpha)Y) &=\sup_{Z \in \U_{\alpha X+ (1-\alpha)Y}}\E_Q[-Z] \\
&\leq \sup_{Z \in \U_{X}\cup \U_Y}\E_Q[-Z] \\
&= \max\{\varphi_Q(X); \varphi_Q(Y) \}.
 \end{align*}

\textit{Monotonicity} and \textit{continuity from above} of $\varphi_Q$ can be proved as in Corollary \ref{cor: dual corollary-Lp}.

b) On the one hand, by Proposition 2.5 and Theorem 2.9 of \citet{frittelli-maggis} and by item a), it follows that
\begin{equation} \label{eq: varphi-repres}
\varphi_Q(X)= \sup_{\tilde{Q} \in \mathcal{P}} R_{\varphi_Q}(\E_{\tilde{Q}}[-X],\tilde{Q}), \quad \mbox{ for any } X \in L^p.
\end{equation}

 On the other hand, by the dual representation of $\rho$ (see \citet{follm-schied} and \citet{fritt-rg}) and by definition of $\rhot$, 
 \begin{align}
 \rhot(X) 
 &= \sup_{Z \in \U_X} \sup_{Q \in \mathcal{P}} \{ \E_Q[-Z]- c_{\rho}(Q) \} \notag \\
 &= \sup_{Q \in \mathcal{P}}   \{ \sup_{Z \in \U_X} \E_Q[-Z]- c_{\rho}(Q) \} \notag\\
  &= \sup_{Q \in \mathcal{P}}   \{ \varphi_Q(X) - c_{\rho}(Q) \}. \label{eq: rhot0}
 \end{align}
Then, \eqref{eq: varphi-repres} and \eqref{eq: rhot0} imply that, for any $X \in L^p$,
\begin{align*}
    \rhot(X) 
     &= \sup_{Q \in \mathcal{P}}   \bigg\{ \sup_{\tilde{Q} \in \mathcal{P}} R_{\varphi_Q}(\E_{\tilde{Q}}[-X],\tilde{Q}) - c_{\rho}(Q) \bigg\} \\
     &= \sup_{Q, \tilde{Q} \in \mathcal{P}}  \big\{  R_{\varphi_Q}(\E_{\tilde{Q}}[-X],\tilde{Q}) - c_{\rho}(Q) \big\}.
\end{align*}
\end{proof}

\section{Acceptance families under robustification/uncertainty}
\label{sec:acceptance_families}

It is well known (see \citet{Delb, follm-schied}) that convex cash-additive risk measures are in a one-to-one correspondence with acceptance sets
\begin{equation*} 
\mathcal{A}_\rho\triangleq\{X\in L^p:\rho(X)\le 0\},
\end{equation*}
formed by all positions that are acceptable and do not require any extra margin.
For risk measures that are not necessarily cash-additive, instead, a single acceptance set defined with respect to the target level $0$ is no more enough but a family of acceptance sets at different target levels is needed. Financially speaking, for general (e.g., quasi-convex or cash-subadditive) risk measures, the notion of acceptability depends on the specified target level (see \citet{drapeau-kupper}).

We recall from \citet{drapeau-kupper}, Theorem 1, that any quasi-convex risk measure $\rho$ is in a one-to-one correspondence with the family $(\mathcal{A}^m_{\rho})_{m \in \mathbb{R}}$ of acceptance sets given by
\begin{equation} 
\label{eq: Am-rho}
\mathcal{A}^m_{\rho}\triangleq \{ X \in L^{p}: \rho(X) \leq m\}, \quad \mbox{ for any } m \in \mathbb{R},    
\end{equation}
via
\begin{equation} \label{eq: rho-Am}
\rho(X)= \inf\{ m \in \mathbb{R}: X \in \mathcal{A}^m_{\rho}\}, \quad \mbox{ for any } X \in L^{p}.
\end{equation}
In other words, each set $\mathcal{A}^m_\rho$ collects the positions that are acceptable at the target level $m$, meaning that the perceived risk of $X$ does not exceed the threshold $m$, and  $\rho(X)$ represents the minimal level $m$ for which $X$ becomes acceptable.
The family of acceptance sets thus provides a richer description of acceptability across different tolerance levels, capturing situations where the admissibility of a position depends on the benchmark or regulatory target. 
\bigskip

In the following result, we provide a relationship between the family $(\mathcal{A}^m_{\rho})_{m \in \mathbb{R}}$ of acceptance sets of the initial risk measure $\rho$ and the family $(\mathcal{A}^m_{\rhot})_{m \in \mathbb{R}}$ of the robust risk measure $\rhot$.

\begin{proposition} 
Let $\U$ be a convex family of uncertainty sets, $\rho$ be a quasi-convex risk measure and $\rhot$ be the corresponding robust risk measure.

Then the following statements hold.

\noindent a) $X \in \A^m_{\rhot} \, \Longleftrightarrow \, \U_X \subseteq \mathcal{A}_{\rho}^m$.

\noindent b) $\rhot(X)= \inf\{m \in \mathbb{R}:  \U_X \subseteq \mathcal{A}_{\rho}^m\}$ for any $X \in L^{p}$.

\noindent c) If $\rho$ is also cash-additive (hence, convex), then $\rhot(X)= \inf\{m \in \mathbb{R}:  \U_X +m \subseteq \mathcal{A}_{\rho}^0\}$ for any $X \in L^{p}$.
\end{proposition}

\begin{proof}
\noindent a) From Proposition \ref{prop: properties robust-cqco}, item d), it follows that $\rhot$ is a quasi-convex risk measure. Hence, by Theorem 1 of \citet{drapeau-kupper}, it is in a one-to-one correspondence with the family $(\mathcal{A}^m_{\rhot})_{m \in \mathbb{R}}$ defined in \eqref{eq: Am-rho}. Furthermore, by definition of $\rhot$, it follows that
\begin{align*}
\A^m_{\rhot} &= \{X \in L^{p}: \rhot(X) \leq m \}\\
&= \{X \in L^{p}: \sup_{Z \in \U_X} \rho(Z) \leq m \}\\
&= \{X \in L^{p}: \rho(Z) \leq m \, \mbox{ for any } Z \in \U_X\}.
\end{align*}

On the one hand, if $X \in \A^m_{\rhot}$ then, by the argument above, $\rho(Z) \leq m$ for any $Z \in \U_X$. Hence, $\U_X \subseteq \mathcal{A}_{\rho}^m$.
On the other hand, if $\U_X \subseteq \mathcal{A}_{\rho}^m$ then
$$
\rhot(X)= \sup_{Z \in \U_X} \rho(Z) \leq  \sup_{Z \in \mathcal{A}_{\rho}^m} \rho(Z) \leq m,
$$
implying that $X \in \A^m_{\rhot}$.

\noindent b) follows immediately from item a) and \eqref{eq: rho-Am}.

\noindent c) If $\rho$ is quasi-convex and cash-additive, then it is also convex (see \citet{CVMMM11,drapeau-kupper,frittelli-maggis}). Hence, by Proposition 2 of \citet{drapeau-kupper}, $\mathcal{A}^m_{\rho}=\mathcal{A}^0_{\rho}-m$ for any $m \in \mathbb{R}$. This and item b) imply that
\begin{align*} 
\rhot(X) &= \inf\{m \in \mathbb{R}:  \U_X \subseteq \mathcal{A}_{\rho}^m\}   \\
&= \inf\{m \in \mathbb{R}:  \U_X \subseteq \mathcal{A}_{\rho}^0-m\}   \\
&= \inf\{m \in \mathbb{R}:  \U_X +m \subseteq \mathcal{A}_{\rho}^0\}
\end{align*}
for any $X \in L^{p}$.
\end{proof}

The following example (on $L^{\infty}$) illustrates the inclusion in item a) of the previous result where the acceptance sets of both $\rho$ and $\rhot$ can be explicitly computed.

\begin{example} 
Fix $\varepsilon, K>0$ and consider, for any $X \in L^{\infty}$,
\begin{align*}
\U_X &=\{ Z \in L^{\infty}: - \varepsilon \leq Z-X \leq \varepsilon\} \\
\rho(X)&=\E[-X] \vee K.
\end{align*}
 It is immediate to check that $\rho$ is a quasi-convex risk measure. Furthermore, the associated robust risk measure reduces to
 \begin{equation*}
 \rhot(X)= \sup_{ X-\varepsilon \leq Z \leq X+\varepsilon} \big(\E[-Z] \vee K \big) 
 =\big(\E[-X]+\varepsilon  \big)\vee K.
 \end{equation*}
It then follows that the family of acceptance sets of $\rho$ is given by
 \begin{equation*}
\mathcal{A}^m_{\rho}= \left\{
\begin{array}{cl}
\{X \in L^{\infty}: \E[-X] \leq m \}
&; \, m \geq K \\
\emptyset&; \, m <K
\end{array}
\right.,
 \end{equation*}
while the family of acceptance sets of $\rhot$ is 
 \begin{equation*}
\mathcal{A}^m_{\rhot}= \left\{
\begin{array}{cl}
\{X \in L^{\infty}: \E[-X]\leq m -\varepsilon\}
&; \, m \geq K \\
\emptyset&; \, m <K
\end{array}
\right.
 \end{equation*}

The explicit formulation of the families above allows to verify directly the inclusion in a) in the previous result. 
For any $m<K$, indeed, we have that both $\mathcal{A}^m_\rho$ and $\A^m_{\tilde{\rho}}$ are empty, so that 
 $X \notin \A^m_{\tilde{\rho}}$ is equivalent to 
 $\U_X \not\subseteq \mathcal{A}^m_\rho$.
 Consider now the case of $m\geq K$. For any $X \in \A^m_{\tilde{\rho}}$ it holds that
 $\mathbb{E}[-X]\leq m-\varepsilon$. Thus, for any $Z \in \U_X$ we have
 $\mathbb{E}[-Z] \leq \mathbb{E}[-X+\varepsilon] = \mathbb{E}[-X] +\varepsilon \leq m$, so
 $Z \in \mathcal{A}^m_\rho$. This shows that:
 $X \in \A^m_{\tilde{\rho}} \Rightarrow \U_X \subseteq \mathcal{A}^m_\rho$. On the other hand, if  $\U_X \subseteq \mathcal{A}^m_\rho$, then $\rho(Z) \leq m $ for any $Z \in \U_X$. In particular, for 
 $Z=X-\varepsilon$ we deduce
 $\rho(X-\varepsilon) \leq m$, which implies $\mathbb{E}[-X]\leq m-\varepsilon$ and thus $X \in \A^m_{\tilde{\rho}}$. This shows that:  $\U_X \subseteq \mathcal{A}^m_\rho \Rightarrow X \in \A^m_{\tilde{\rho}}$. 
\end{example}

\section{Applications and examples}
\label{sec: appl car-examples}

In this section, we provide two applications of the robustification: the former to capital allocations, the latter to the law invariant case with uncertainty sets based on the Wasserstein distance.
 Some illustrative examples of families of uncertainty sets and robust risk measures are also given.

\subsection{Capital allocation under robustification /ambiguity}

We start recalling the main features of capital allocation problems in order to provide an application to the robustness issue.

Roughly speaking, a capital allocation rule replies to the question of how to share the capital requirement for an aggregate risky position (prescribed by a given risk measure $\rho$) among its different sub-units. Here below, we recall the classical definition of a capital allocation rule.

\begin{definition}[see \citet{kalk}]
 Given a risk measure $\rho : L^{p} \to [-\infty,+\infty]$, a \textit{capital allocation rule (CAR)} for $\rho$ is a map $\Lambda:L^{p}\times L^{p} \to [-\infty,+\infty]$ such that $\Lambda(Y,Y)=\rho(Y)$  for every $Y \in L^{p}$.
\end{definition}

Notice that a CAR $\Lambda$ is defined for any pair $(X,Y)\in L^{p}\times L^{p}$, where $Y$ is always interpreted as an aggregate position (or portfolio) and $X$ as a sub-unit (or sub-portfolio).
The condition $\Lambda(Y,Y)=\rho(Y)$ means that the capital allocated to $Y$ when considered as a stand-alone portfolio is exactly the margin $\rho(Y)$. Moreover,
\citet{kalk} defines the following desirable properties for a capital allocation rule $\Lambda$ associated to a risk measure $\rho$:

\begin{itemize}
\item[-]{\it Full allocation}: 
for any $(Y_i)_{i=1,...,n} \subseteq L^p,Y \in L^p$ with $\sum_{i=1}^n Y_i=Y$, then $\sum_{i=1}^n {\Lambda}(Y_i,Y) = {\Lambda}(Y,Y)$.

\item[-] {\it No-undercut}: $\Lambda(X,Y)\leq \rho(X)$ for every  $X,Y\in L^{p}$.
\end{itemize}

Full allocation is satisfied by the most popular capital allocation rules at least for the coherent case, and it implies that the whole risk capital is entirely split among the sub-units of a risky position. 
Since $\Lambda(X,Y)$ represents the capital allocated to the sub-portfolio $X$ to hedge the total portfolio $Y$, the no-undercut axiom requires that the capital allocated to $X$ as a sub-unit of $Y$ does not exceed the margin required for $X$. In the terminology of \citet{tsanakas}, no-undercut corresponds to the {\it non-split} requirement of $X$ from $Y$,  while from a game theory perspective, it corresponds to the core property of \citet{denault}.

Although full allocation and no-undercut are somehow incompatible when dealing with the non-coherent case, capital allocation can still be performed for general risk measures without losing soundness (see, e.g., \citet{CRG, Canna-Ce-Ro-CAR21}). 

When ambiguity is introduced, both the evaluation of total risk and the capital allocation are affected. To capture this effect, we extend the notion of capital allocation to
the robust framework by defining the ``robustified" capital allocation rule $\tilde{\Lambda}$ associated with the robust risk measure $\tilde\rho$.

\begin{definition}
Let $\U$ be a family of uncertainty sets and $\Lambda$ be a CAR for a risk measure $\rho$.

We define the robustified version of $\Lambda$ as a map $\tilde{\Lambda}:L^p \times L^p \to [-\infty; + \infty]$ given by
\begin{equation} \label{eq: Lambda-tilde}
\tilde{\Lambda}(X,Y) \triangleq \sup_{Z \in \U_X} \Lambda(Z,Y), \quad \mbox{ for any } X, Y \in L^p.
\end{equation}
\end{definition}

The properties of $\tilde{\Lambda}$ as a CAR for the robust risk measure $\rhot$ are investigated here below. As expected, the property of full allocation is not preserved under robustification. No-undercut, instead, is fulfilled. 

\begin{proposition} 
Let $\Lambda$ be a CAR for a risk measure $\rho$ and assume that the family $\U=(\U_X)_{X \in L^{p}}$ of uncertainty sets is such that $X \in \U_X$ for any $X \in L^p$. 

a) If $\Lambda$ satisfies no-undercut, then the robustified $\tilde{\Lambda}$ defined in \eqref{eq: Lambda-tilde} satisfies no-undercut and
\begin{equation}
\label{eq: no-und}
\rho(Y) \leq \tilde{\Lambda}(Y,Y) \leq \rhot(Y), \mbox{ for any } Y \in L^p.
\end{equation}

\noindent b) If $\cup_{i=1}^n \U_{Y_i} \supseteq \U_Y$ for any $(Y_i)_{i=1,...,n} \subseteq L^p$ and $Y \in L^p$ with $\sum_{i=1}^n Y_i=Y$,
then $\tilde{\Lambda}(Y,Y) \leq \max_{i=1,...,n} \tilde{\Lambda}(Y_i,Y)$.

If, in addition, $0 \in \U_{Y_i}$ for any $i=1,...,n$ and $\Lambda(0,Y) \geq 0$, then 
$\tilde{\Lambda}(Y,Y) \leq \sum_{i=1}^n \tilde{\Lambda}(Y_i,Y)$.
\end{proposition}

\begin{proof}
a) The no-undercut of $\Lambda$ implies the no-undercut of $\tilde{\Lambda}$. Indeed, for any $X,Y \in L^p$,
$$
\tilde{\Lambda}(X,Y) =\sup_{Z \in \U_X} \Lambda(Z,Y) \leq \sup_{Z \in \U_X} \rho(Z)= \rhot(X).
$$
In particular, $\tilde{\Lambda}(Y,Y) \leq \rhot(Y)$.

It remains to prove the first inequality in \eqref{eq: no-und}. Since $Y \in \U_Y$ for any $Y \in L^p$, it holds that
$$
\tilde{\Lambda}(Y,Y) =\sup_{Z \in \U_Y} \Lambda(Z,Y) \geq \Lambda(Y,Y)=\rho(Y),
$$
where the last equality follows from $\Lambda$ being a CAR for $\rho$.
This completes the proof of item a).

\noindent b) Consider any $(Y_i)_{i=1,...,n} \subseteq L^p$ and $Y \in L^p$ with $\sum_{i=1}^n Y_i=Y$. By definition of $\tilde{\Lambda}$ and by the properties on the uncertainty sets,
\begin{align}
\tilde{\Lambda}(Y,Y) &= \sup_{Z \in \U_{Y} }\Lambda(Z,Y) \notag \\
&\leq \sup_{Z \in \cup_{i=1}^n \U_{Y_i} }\Lambda(Z,Y) \label{eq: ineq-L} \\
&\leq \max_{i=1,...n } \sup_{Z \in \U_{Y_i}} \Lambda(Z_i,Y). \notag
\end{align}

Assume now that $0 \in \U_{Y_i}$ for any $i=1,...,n$ and $\Lambda(0,Y) \geq 0$.
This implies that, for any $i=1,..,n$,
\begin{equation} \label{eq: L-positive}
\sup_{Z \in \U_{Y_i}} \Lambda(Z_i,Y) \geq \Lambda(0,Y) \geq 0.
\end{equation}
From \eqref{eq: ineq-L} and \eqref{eq: L-positive}, it then follows that
\begin{align*}
\tilde{\Lambda}(Y,Y) & \leq \sup_{Z \in \cup_{i=1}^n \U_{Y_i} }\Lambda(Z,Y) \\
&\leq \sum_{i=1}^n \sup_{Z \in \U_{Y_i}} \Lambda(Z,Y) \\
&=\sum_{i=1}^n \tilde{\Lambda}(Y_i,Y).
\end{align*}
\end{proof}

The previous result underlines that, in general, $\tilde{\Lambda}$ is not necessarily a CAR but only a sub-CAR, meaning that $\tilde{\Lambda}(Y,Y)$ does not necessarily coincide with $\rhot(Y)$ but $\tilde{\Lambda}(Y,Y) \leq \rhot(Y)$. 
See \citet{CRG} for a detailed discussion on sub-CARs.
In particular, for a stand-alone portfolio $(Y,Y)$, the robustified $\tilde{\Lambda}$ requires to allocate a greater amount than the margin evaluated with the non-robust $\rho(Y)$ but a lower amount than $\rhot(Y)$, that is, the margin prescribed by the robust $\rhot$. Loosely speaking, the proposed $\tilde{\Lambda}$ takes into account ambiguity (by ``overloading'' $\rho(Y)$) in a moderate way, that is, without requiring the allocation of the whole robust capital $\rhot(Y)$.

As expected, also full allocation fails to be satisfied for the robustified version of $\Lambda$. This fact is not surprising at all also reminding that no-undercut and full allocation do not hold together for CARs of general risk measures (see the discussion in \citet{CRG}). However, under the additional assumptions that null positions belong to all $\U_{Y_i}$ and that $\Lambda (0,Y) \geq 0$, a weaker property than full allocation is fulfilled. This property means that the sum of the capital allocated to any sub-unit $Y_i$ is greater than the total capital allocated to the aggregate position $Y$. Financially speaking, $\sum_{i=1}^n \tilde{\Lambda}(Y_i,Y) - \tilde{\Lambda}(Y,Y) \geq 0$ can be interpreted as a ``security" extra capital to be allocated also in view of ambiguity/robustification. Note that hypothesis $\Lambda(0,Y) = 0$ (or, more generally, $\Lambda(0,Y) \geq 0$) is reasonable and quite commonly assumed in the literature. Indeed, this means that the capital allocated to a null position is zero (or, more generally, that even ``doing nothing” may require capital).

\subsection{Law invariant case under Wasserstein distance}

We apply now the robustification to the case of law invariant risk measures and uncertainty sets based on the Wasserstein distance. 
A notable example of a family of uncertainty sets is, indeed, given by closed balls, centered at $X \in L^p$ with respect to a suitable metric and with a fixed radius $\epsilon > 0$. Specifically, one can consider 
the Wasserstein distance of order $p$ (see \citet{Villani21} for a comprehensive discussion of this metric), defined as
\[
d_{W_p}(X, Y) \triangleq \left( \int_0^1 |F_X(u) - F_Y(u)|^p \, du \right)^{1/p}, \quad \mbox{ for } p \in [1, +\infty).
\]
For $p = +\infty$, one can set $d_{W_\infty}(X, Y) = \lim_{p \to \infty} d_{W_p}(X, Y)$.
 Consider now uncertainty sets of the following form:
\begin{equation} \label{eq: wasser ball}
    \mathcal{U}_X = \{Z \in L^p : d_{W_p}(X, Y) \leq \epsilon\}, \quad \mbox{ for any } X \in L^p.
\end{equation}

Like closed balls, these sets fit directly into our framework and have the additional property of being convex. Moreover, it is straightforward to verify that they generate a law-invariant, closed, order preserving and convex family.

In the following result, we provide an estimate on $\tilde{\rho}$ when uncertainty sets are based on the Wasserstein distance. This result is inspired by Theorem 2 of \citet{righi-arxiv}. Compared to the result of this author (holding for convex cash-additive risk measures), ours is formulated for quasi-convex and cash-subadditive risk measures. Due to this generalization, we are only able to prove an estimate of $\rhot$.

\begin{proposition}
    Let $p\in[1, +\infty]$ and  $\rho$ be a quasi-convex and cash-subadditive law invariant risk measure with dual representation \eqref{eq: repr rho} with the supremum attained, i.e. $\rho(X)= \max_{Q\in \mathcal{P} }
    R_\rho(\mathbb{E}_Q[-X],Q)$. 
    
If the uncertainty sets $\U_X$ are chosen as in \eqref{eq: wasser ball}, then 
$$
 \tilde{\rho}(X) \leq \rho(X) + \epsilon \Bigg\| \frac{ d {Q}_X^*}{d P}\Bigg\|_q, \quad \mbox{ for any } X \in L^p,
$$ 
where ${Q}_X^* \in \argmax R_\rho(\mathbb{E}_Q[-X],Q)$.
\end{proposition}

\begin{proof}
Let $X \in L^p$ be arbitrarily fixed and let ${Q}_X^*\in \argmax R_\rho(\mathbb{E}_Q[-X],Q)$. 

Following \citet{righi-arxiv},
we denote by 
$$
    f_{Q_X^*}(Z)= \sup_{Z^\prime \sim Z} \mathbb{E}_{Q_X^*}[Z^\prime]
       =  \int_0^1 F^{-1}_{Z}(u) F^{-1}_{{\rm d}Q_X^*/{\rm d}P}(u) {\rm d}u, \quad \mbox{ for any } Z \in L^p.
    $$ 
Then, by \citet[Theorem 2]{righi-arxiv}, 
\begin{equation} \label{eq: ineq phi}
  \varphi_{Q_X^*}(X)= \sup_{Z \in \U_X} f_{Q_X^*}(-Z) \leq f_{Q_X^*}(-X) + \epsilon \left\| \frac{dQ_X^*}{dP} \right\|_q.
\end{equation}

From the definition  of $\tilde{\rho}$ and the dual representation of $\rho$,
it holds: 
    \begin{align}
        \tilde{\rho}(X) & = 
        \sup_{Z\in \mathcal{U}_X} \max_{Q \in \mathcal{P}} R_
        \rho\left(\mathbb{E}_Q[-X],Q\right) \notag\\
        &= \sup_{Z\in \mathcal{U}_X} R_\rho\left(\mathbb{E}_{Q_X^*}[-X],{Q_X^*}\right) \notag \\
        & = R_\rho\left(\sup_{Z \in \mathcal{U}_X}\mathbb{E}_{{Q_X^*}}[-Z],{Q_X^*}\right) \label{eq: R-sup}\\
        &= R_\rho\left(\varphi_{Q_X^*}(X) ,Q_X^*\right), \label{eq: 000R}
    \end{align}
where \eqref{eq: R-sup} follows from the non-expansivity of $R_\rho$ (as in the proof of Theorem \ref{thm: dual represen-qco rho-Lp}).
Applying \eqref{eq: ineq phi} and increasing monotonicity and non-expansivity of $R_\rho$ to \eqref{eq: 000R}, we obtain
\begin{align}
    \tilde{\rho}(X) &\leq  R_\rho \left( f_{Q_X^*}(-X)+ \epsilon \left\| \frac{dQ_X^*}{dP} \right\|_q, Q_X^*\right) \notag \\
    & \leq 
    R_\rho \left( f_{Q_X^*}(-X), Q_X^*\right)+ \epsilon \left\| \frac{dQ_X^*}{dP} \right\|_q \notag\\
    & = 
   \sup_{X^\prime \sim X} R_\rho \left(  \mathbb{E}_{Q_X^*}[-X^\prime], Q_X^*\right)+ \epsilon \left\| \frac{dQ_X^*}{dP} \right\|_q \label{eq: R-star}\\
   & \leq  
   \sup_{X^\prime \sim X} \rho(X^\prime)+ \epsilon \left\| \frac{dQ_X^*}{dP} \right\|_q \label{eq: R-star1} \\
    & = \rho(X) + \epsilon \left\| \frac{dQ_X^*}{dP} \right\|_q, \label{eq: R-star2}
\end{align}
where \eqref{eq: R-star} is due to non-expansivity of $R_\rho$ (see the proof of Theorem \ref{thm: dual represen-qco rho-Lp}), \eqref{eq: R-star1} to the dual representation of $\rho$, and \eqref{eq: R-star2} to law invariance of $\rho$.
\end{proof}

\subsection{Examples of families of uncertainty sets} 
\label{sec: examples}

Here below, we provide some examples of c-quasi-convex (and solid) families $\mathcal{U}$ of uncertainty sets built on quasi-convex and cash-subadditive risk measures.
\smallskip 

\begin{example} 
\label{ex: qco}
Given a quasi-convex and cash-subadditive risk measure $\rho_1$, consider
\begin{equation*}
\mathcal{U}_X =\left\{ Z \in L^\infty: |\rho_1(Z)-\rho_1(X)| \leq \varepsilon\right\}, \quad \mbox{ for any } X \in L^{\infty}. 
\end{equation*}

We are now going to show that the family $\mathcal{U}=(\U_X)_{X \in L^\infty}$ 
is c-quasi-convex.

Let $X,Y \in L^{\infty}$ and $\lambda \in [0,1]$ be fixed arbitrarily. 
Without loss of generality, we assume that $\rho_1(X) \geq \rho_1(Y)$. 
From the definition of $\mathcal{U}$ and quasi-convexity of $\rho_1$, it follows that, for any $Z\in \mathcal{U}_{\lambda X + (1-\lambda)Y}$,
\begin{equation}
\label{eq: ex U1}
\rho_1(Z) \leq \rho_1(\lambda X + (1-\lambda) Y)+ \varepsilon \leq \max \left\{ \rho_1(X),\rho_1(Y)\right\}+ \varepsilon=\rho_1(X)+ \varepsilon.
\end{equation}

If $\rho_1(Z) \geq \rho_1(X) - \varepsilon$, then \eqref{eq: ex U1} implies that $Z\in \mathcal{U}_X$, thus $ Z\in \mathcal{U}_X \subseteq \mathcal{U}_X \cup \mathcal{U}_Y +L^\infty_+$.

If $\rho_1(Z) < \rho_1(X) - \varepsilon$, instead, we need to verify that 
 there exist $K^*\in L^\infty_+$ and $Z^* \in \mathcal{U}_X \cup \mathcal{U}_Y$ such that $Z=Z^* + K^*$. In particular, it would be sufficient to prove that there exist 
 $k^* \in \mathbb{R}_+$ such that the random variable $Z^* \triangleq Z - k^*$ belongs to $\mathcal{U}_X$. Indeed, in this case, 
    \[
    Z = Z^* + k^* \in \mathcal{U}_X + L^\infty_+ \subseteq \mathcal{U}_X \cup \mathcal{U}_Y + L^\infty_+.
    \]
 To find a suitable $k^*$, define the function $f: \mathbb{R}_+ \to \mathbb{R}$ by:
$$
f(k) = \rho_1(Z-k), \quad \mbox{ for any } k \in \mathbb{R}_+.$$
It follows immediately that $f$ is increasing (by monotonicity of $\rho_1$) and $f(0) = \rho_1(Z)$. Furthermore, cash-subadditivity of $\rho_1$ (hence $L^\infty$-continuity of $\rho_1$ by Proposition 2.1, b), of \citet{CVMMM11}) implies continuity of $f$ on $\mathbb{R}_+$. 
    
    Set now:
    \begin{equation*}
    \bar{k} \triangleq \rho_1(X) + \varepsilon - \rho_1(Z) > 0.
    \end{equation*}
    By cash-subadditivity of $\rho_1$, it then follows that
    \begin{equation*}
    f(\bar{k}) = \rho_1(Z - \bar{k}) \geq \rho_1(Z) + \bar{k}= \rho_1(X) + \varepsilon.
    \end{equation*}
    Since $f$ is continuous on $[0, \bar{k}]$ and its image contains the interval $[\rho_1(Z), \rho_1(X) + \varepsilon]$, by the Intermediate Value Theorem there exists $k^* \in [0, \bar{k}]$ such that:
     \begin{equation*}
    f(k^*) \in [\rho_1(X) - \varepsilon, \rho_1(X) + \varepsilon],
     \end{equation*}
    i.e., $\rho_1(Z-k^*)\in [\rho_1(X) - \varepsilon, \rho_1(X) + \varepsilon] $. So, $Z^* = Z - k^* \in \mathcal{U}_X$.
    
     By the above arguments, the family $\U$ is c-quasi-convex (but it fails to satisfy monotonicity).
\end{example}

\begin{example}\label{ex: solid-mon}
We now slightly modify Example \ref{ex: qco} in order to obtain a family $\U$ also satisfying solidity and monotonicity.

Set
\begin{equation}
\label{eq: ex U2}
    \mathcal{U}_X = \left\{ Z \in L^\infty: \rho_1(Z)\leq \rho_1(X)+ \varepsilon\right\}, \quad \mbox{ for any } X \in L^{\infty}.
\end{equation}

c-quasi-convexity can be verified as in Example \ref{ex: qco}.
Solidity is an immediate consequence of monotonicity of $\rho_1$.
Regarding monotonicity, let $X, Y \in L^\infty$ with $X \leq Y$. If $Z \in \mathcal{U}_Y$, then, by monotonicity of $\rho_1$,
\begin{equation*}
\rho_1(Z) \leq \rho_1(Y) + \varepsilon \leq \rho_1(X) + \varepsilon,
\end{equation*}
so $Z \in \mathcal{U}_X$, and thus
$\mathcal{U}_Y \subseteq \mathcal{U}_X$.

To sum up, the family $\U$ in \eqref{eq: ex U2} satisfies c-quasi-convexity, solidity and monotonicity, hence, by Lemma \ref{lem: c-qco}, ii), also quasi-convexity. 
\end{example}

\begin{remark}
To be more concrete, two relevant cases of risk measures $\rho_1$ falling within the setting of Examples \ref{ex: qco} and \ref{ex: solid-mon} are, for instance,
those taking into account ambiguity with respect to discount factors and those based on generalized entropy.
\smallskip

\emph{Ambiguity with respect to discount factors.}  
Example 7.2 in \citet{ELK-R} covers the case of a regulator who is evaluating the riskiness of financial positions under ambiguity about the (stochastic) discount factor \( D \), which she/he believes lies within a known range $[d_{\rm min}, d_{\rm max}]$ with $d_{\rm min} , d_{\rm max} \in [0,1]$ being two constants.
If the initial (spot) risk measure is \( \rho_0 \), then the risk measure $\rho_1$ that takes into account ambiguity on the discount factor via worst-case scenarios, would be defined as
$$
\rho_1(X)= \sup_{0\leq d_{\rm min} \leq D \leq d_{\rm max}\leq 1 } \rho_0(D X), \quad \mbox{ for any } X \in L^{\infty}.
$$
See \citet{ELK-R} for a detailed treatment.

\smallskip

\emph{Generalized entropy.}
q-entropic risk measures on losses are introduced in \cite{DNRG} to provide a family of cash non-additive (convex) risk measures generalizing entropic ones and, in a dynamic setting and once modified, fulfilling the so-called horizon longevity. Namely, a q-entropic risk measure on losses is defined as
\begin{equation} \label{eq: q-entropy}
      \rho_1(X)= \ln_q \mathbb{E}[\exp_q((X+\beta)^-)], \mbox{ for } q \in (0,1),
\end{equation}
where $\beta >0$ is a fixed target, $\ln_q$ and $\exp_q$ generalize the the exponential and logarithmic functions\footnote{We recall from \citet{Tsallis88,Tsallis09} that
\begin{align*}
\exp_q(x) &\triangleq \big[1 + (1 - q)x\big]^{\frac{1}{1 - q}}
\begin{cases}
\text{ defined for any } x \geq \dfrac{1}{q - 1}\; \text{ when } q \in (0,1) \\
\text{ defined for any } x < \dfrac{1}{q - 1}\;  \text{ when } q > 1
\end{cases}
 \\
\ln_q(x) &\triangleq \dfrac{x^{1 - q} - 1}{1 - q}
\begin{cases}
\text{ defined for any } x \geq 0\;  \text{ when } q \in (0,1) \\
\text{ defined for any } x > 0\;  \text{ when } q > 1
\end{cases}.
\end{align*}}, and $\ln_q \mathbb{E}[\exp_q(\cdot)]$ is the generalized entropic risk measure (see \citet{Tsallis88,Tsallis09}). 

As argued in \citet{MaTian21}, compared with the entropic risk measure, the  generalized entropic risk measure does not satisfy the cash-additivity but satisfies convexity and cash-subadditivity for any $q\in (0,1)$. 
The definition of q-entropic risk measure as in \eqref{eq: q-entropy} also guarantees that it can be defined for all risky positions and that is convex everywhere (see \cite{DNRG}).
\end{remark}

\section{Conclusions}
\label{sec:conclusions}
This paper develops a comprehensive framework for the robustification of risk measures beyond the classical convex and cash-additive setting. We introduced two complementary approaches for constructing robust quasi-convex risk measures: one based on a quasi-convex base functional and convex uncertainty sets, and another relying on quasi-convex (or c-quasi-convex) uncertainty families applied to a general risk measure. These two schemes clarify how the geometry of the uncertainty sets and the regularity of the base risk measure jointly determine the properties of the robust measure.

Within this framework, we established general results on the preservation of monotonicity, (quasi-)convexity, law invariance, and continuity from above under robustification. We also characterized the largest, or consolidated,
uncertainty family associated with a given robust risk measure, extending earlier results from the convex and cash-additive literature. Building on the dual representations of \citet{CVMMM11} and \citet{frittelli-maggis}, we derived penalty-type dual forms for robust quasi-convex and cash-subadditive risk measures and show that the classical convex, cash-additive case emerges as a special instance of our general formulation.

Further, we analyzed acceptance families under robustification, establishing the relation between the level sets of the base and robust risk measures and demonstrate how robustness can be interpreted as a systematic enlargement of acceptance sets to account for model ambiguity. The theoretical results were illustrated through representative examples based on Wasserstein distance and $p$-norm balls and c-quasi-convex uncertainty families inspired by set-valued analysis. Finally, we extended the analysis to capital allocation rules, providing a consistent approach under ambiguity.

The proposed framework unifies and extends existing results on robust convex risk measures, offering a general and flexible foundation for risk measurement and allocation in the presence of frictions, illiquidity, or discounting ambiguity. Future research may address dynamic extensions, stochastic dominance consistency, or empirical applications to portfolio selection and stress testing under ambiguity.

\end{document}